\date{October 26, 2016}

\documentclass[12pt]{amsart}
\usepackage{latexsym,amsmath,amsfonts,amscd,amssymb}
\usepackage{graphicx}
\newtheorem{theorem}{Theorem}[section]
\newtheorem{theorem*}{Theorem}
\newtheorem{lemma}[theorem]{Lemma}
\newtheorem{proposition}[theorem]{Proposition}
\newtheorem{corollary}[theorem]{Corollary}

\theoremstyle{remark}

\newcommand{\EE}{{\mathbb E}}

\newcommand{\NN}{{\mathbb N}}
\newcommand{\PP}{{\mathbb P}}

\newcommand{\RR}{{\mathbb R}}

\newcommand{\bdN}{{\boldsymbol N}}
\newcommand{\bdS}{{\boldsymbol S}}
\newcommand{\bdT}{{\boldsymbol T}}
\newcommand{\bdX}{{\boldsymbol X}}
\newcommand{\bdkappa}{{\boldsymbol \kappa}}

\renewcommand{\a}{\alpha}

\title{Double spend races}

\date{\tiny February 9th 2017}

\author[C. Grunspan]{Cyril Grunspan}
\address{Cyril Grunspan\newline{}\indent L\'eonard de Vinci P\^ole Univ, Research Center, Labex R\'efi
\newline{}\indent 92 916 Paris-La D\'efense, France}
\email{cyril.grunspan@devinci.fr}

\author[R. P\'{e}rez-Marco]{Ricardo P\'{e}rez-Marco}
\address{Ricardo P\'{e}rez-Marco\newline{}\indent CNRS, IMJ-PRG, Labex R\'efi
, Labex MME-DDII\newline{}\indent B\^at. Sophie Germain, Case 7012, 75205-Paris Cedex 13, France}
\email{ricardo.perez.marco@gmail.com}

\address{\tiny Author's Bitcoin Beer Address (ABBA)\footnote{\tiny Send some bitcoins to support our research at the pub.}:\newline{}\indent 1KrqVxqQFyUY9WuWcR5EHGVvhCS841LPLn} 
\address{\includegraphics[scale=0.34]{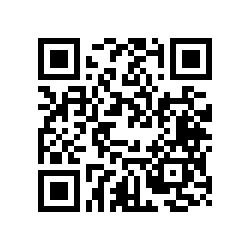}}

\begin{document}

\begin{abstract}
We correct the double spend race analysis given in Nakamoto's foundational Bitcoin article and give 
a closed-form formula for the probability of success of a double spend attack using
the regularized incomplete beta function. We give a proof of the exponential decay on the number of 
confirmations, often cited in the literature, and find an asymptotic formula. Larger number of confirmations are necessary compared to those given by Nakamoto. 
We also compute the probability conditional to the known validation time of the blocks. This provides a finer 
risk analysis than the classical one.
\end{abstract}

\maketitle

\noindent 
\emph{\footnotesize To the memory of our beloved teacher Andr\'e Warusfel who taught us how to 
have fun with the applications of mathematics.}

\section{Introduction.}

The main breakthrough in \cite{N} is the solution to the \textit{double spend problem}. Before this discovery no one knew how to avoid 
the double spending of an electronic currency unit without the supervision of a central authority. This made Bitcoin the first 
form of \textit{peer-to-peer} (P2P) electronic currency. 

A double spend attack can only be attempted with a substantial fraction of the 
hashrate used in the \textit{proof-of-work} of the Bitcoin network. The attackers will start a \textit{double spend race}
against the rest of the network to replace the last blocks of the blockchain by secretly mining an alternate blockchain. 
The last section of the Bitcoin's white paper \cite{N} computes the probability that the attackers catch up. However Nakamoto's analysis is not accurate since he makes the simplifying assumption that honest miners validate blocks at the expected rate. We present a correct analysis and give a closed-form formula for the exact probability.

\pagebreak

\begin{theorem*}
 Let $0<q<1/2$, respectively $p=1-q$, be the relative hash power of the group of the attackers, respectively of honest miners. 
 After $z$ blocks have been validated by the honest miners, the probability of success of the attackers is 
 $$
 P(z)=I_{4pq}(z,1/2) \ 
 $$
 where $I_x(a,b)$ is the regularized incomplete beta function
 $$
 I_x(a,b)=\frac{\Gamma(a+b)}{\Gamma(a) \Gamma(b)} \int_0^x t^{a-1}(1-t)^{b-1} \, dt \ .
 $$
\end{theorem*}

In general, for $z\geq 2$, these probabilities $P(z)$ are larger than those $P_{SN}(z)$ obtained by Nakamoto. 
From the standpoint of bitcoin security, this shows than larger confirmation times $z$ are necessary compared to those $z_{SN}$ 
given by Nakamoto, in particular this happens when the share of hashrate $q$ of the attackers is important. The following 
table shows the number $z$ of confirmations to wait compared to those $z_{SN}$ given by Nakamoto for an attacking hashrate 
of $10\%$ (or $q=0.1$) and a probability of success of the attackers less than $0.1\%$.

\medskip

$$
\begin{array}{|c||c|c|c|c|c|c|c|c|}
 \hline
 q & 0.10 & 0.15 & 0.20 & 0.25 & 0.30 & 0.35 & 0.40 & 0.45 \\ 
 \hline
 z & 6 & 9 & 13 & 20 & 32 & 58 & 133 & 539\\ 
 \hline
 z_{SN} & 5 & 8 & 11 & 15 & 24 & 41 & 81 & 340 \\ 
 \hline
\end{array}
$$

\smallskip

{\centerline {\textbf {\footnotesize{Table 1. Comparison of number of confirmations.}}} 

\medskip

Nakamoto claims in \cite{N} that the probability $P(z)$ converges  
exponentially to $0$ with $z$. This result is intuitively expected and 
cited at large but there is no proof available in the literature. We give here a rigorous proof of this result. More 
precisely we give precise asymptotics both for $P_{SN}(z)$ and $P(z)$ showing the exponential decay.

\begin{theorem*}
 When $z\to +\infty$ we have , with $s=4pq<1$,
$$
P(z) \sim \frac{s^{z}}{\sqrt{\pi (1-s) z}} \ .
$$
and with $\lambda =q/p$, and $c(\lambda) =\lambda -1-\log \lambda >0$, 
$$
P_{SN}(z)\sim \frac{e^{-z c(\lambda)}}{2}
$$
\end{theorem*}
We can check that $-\log s >c(\lambda)$ which means that $P_{SN}(z) \prec P(z)$ for $z$ large.

\pagebreak

\textbf{A finer risk analysis.}

\medskip

We analyze a new parameter in the risk of a double spend. The probability of success of the attackers increases with the 
time $\tau_1$ it takes to validate the $z$ transactions since they have more time to secretly mine their alternate blockchain. 
On the other hand the task of the attackers is more difficult if the validations happen faster than the expected time. 
The value of $\tau_1$ is known, therefore what is really relevant is the 
conditional probability assuming $\tau_1$ is known. We introduce the dimensionless parameter $\kappa$ which measures 
the deviation from average time:
$$
\kappa=\frac{\tau_1}{zt_0} \ ,
$$
where $t_0$ is the average time of block validation by honest miners ($t_0= \tau_0 /p$, where $\tau_0= 10 \, {\hbox{\rm min}}$ for the Bitcoin network).

\begin{center}
\includegraphics[scale=0.45]{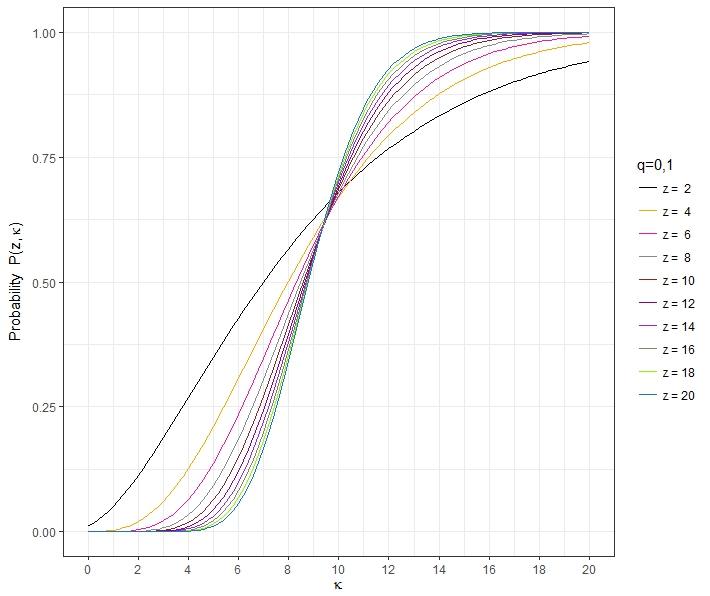} \\
{\footnotesize Figure 1. Probability of success as a function of $\kappa$ with $q=0.1$}
\end{center}


We study the probability $P(z, \kappa)$ of success of the attackers. We can recover the previous probabilities with the  
$P(z,\kappa )$, $0<\kappa <1$.

\begin{theorem*}
We have
$$
P_{SN}(z)=P(z,1) \ ,
$$ 
and 
$$
P(z)=\int_0^{+\infty } P(z, \kappa) \, d\rho_z(\kappa) \ ,
$$
with the density function
$$
d\rho_z(\kappa) = \frac{z^z}{(z-1)!}\kappa^{z-1} e^{-z\kappa}\, d\kappa \ .
$$ 
\end{theorem*}

We give a closed-form formula for $P(z,\kappa)$. 

\begin{theorem*}
We have
$$
P(z, \kappa)=1-Q(z, \kappa z q/p) + \left (\frac{q}{p}\right )^{z} e^{\kappa z\frac{p-q}{p}} Q(z, \kappa z) \ .
$$
\end{theorem*}

Here $Q$ denotes the incomplete gamma function
$$
Q(s,x)=\frac{\Gamma(s,x)}{\Gamma(x)} \ ,
$$
where 
$$
\Gamma(s,x)=\int_x^{+\infty} t^{s-1}e^{-t} \, dt 
$$

We find also the asymptotics for $z\to +\infty$ for different values of $\kappa$. 

\begin{theorem*} 
The following hold for $z\to +\infty$,
\begin{enumerate}
  \item For $0<\kappa<1$, 
  $$
  P(z,\kappa) \sim \frac{1}{1-\kappa \lambda} \frac{1}{\sqrt{2\pi z}}
  e^{-zc(\kappa \lambda)} \ .
  $$
  \item For $\kappa=1$, 
  $$
  P(z,1)=P_{SN}(z)\sim \frac{1}{2} e^{-zc(\lambda)} \ .
  $$
  \item For $1<\kappa<p/q$, 
  $$
  P(z,\kappa)\sim \frac{\kappa (1-\lambda)}{(\kappa -1)(1-\kappa \lambda)} 
  \frac{1}{\sqrt{2\pi z}} e^{-zc(\kappa \lambda)} \ .
  $$
  \item For $\kappa=p/q$, $P(z, p/q)\to 1/2$ and 
  $$
  P(z, p/q)-1/2 \sim \frac{1}{2\pi z} \left (\frac13+\frac{q}{p-q}\right ) \ .
  $$
  \item For $p/q<\kappa$, $P(z,\kappa) \to 1$ and
  $$
  1-P(z,\kappa)\sim \frac{\kappa (1-\lambda)}{(\kappa -1)(\kappa \lambda-1)} 
  \frac{1}{\sqrt{2\pi z}} e^{-zc(\kappa \lambda)} \ .
  $$
\end{enumerate}
\end{theorem*}

Using a concavity argument we show that 
$P(1) \leq P_{SN}(1)$, but in general, for $z\geq z_0$, we have $P_{SN}(z)\leq P(z)$. We do compute an explicit, non sharp, 
value $z_0$ for which this inequality holds:

\begin{theorem*}
  Let $z \in \NN$. A sufficient condition for having
  $P_{SN} (z) < P ( z)$ is $z \geq z_0$ with $z_0=\lceil z_0^*\rceil$ being the smallest integer greater or equal
  to
$$ 
 z_0^* = \max \left( \frac{2}{\pi ( 1 -
     q/p )^2}\, , \, \frac{1}{2 \sqrt{2}} - \frac{\left( 1 +
     \frac{1}{\sqrt{2}} \right)}{2}  \frac{\log \left( \frac{2 \psi (
     p)}{\pi}  \right)}{\psi ( p)} \right) 
$$
  where $\psi ( p) = \frac{q}{p} - 1 - \log \left(  \frac{q}{p}
  \right) - \log \left( \frac{1}{4 pq} \right) > 0$.
\end{theorem*}

We also provide a double entry tables of the $P(z,\kappa )$ 
for different values of $(q,\kappa )$ for $z=3$ and $z=6$.  For a complete set of tables for $z=1,2,\ldots 9$ of practical use we refer to the companion article \cite{GPM}.

\section{Mathematics of mining.}

We review some basic results in probability (see \cite{F} vol.2, p.8) 
and of Bitcoin mining (see \cite{PM} for an overview of Bitcoin protocol).

A hashing algorithm digest any file into a fixed length string of bits. 
The slightiest modification of the original file produces a completely different output.
The bits of the output appear with a random frequency and it is computationnally hard to find 
collisions (different inputs yielding the same output).
Hashing algorithms are used for example to check the integrity and non-tampering of files. 

The two main hashing algorithms used in the Bitcoin protocol are RIPEMD-160 and SHA-256 that produce 
outputs of $160$ and $256$ bits respectively. The mining algorithm 
consists in performing the double SHA-256 of the block header (doubled to prevent ``padding attacks''). 

The consensus protocol and security in the Bitcoin network relies on the process of bitcoin mining and validating transactions.
It consists on the iteration of computation of block header hashes changing a nonce 
\footnote{More precisely, a double hash SHA256(SHA256(header)) is computed, changing a nonce and an extra-nonce.} in order to find a 
hash below a predefined threshold, the \textit{difficulty} \cite {N}. For each new hash 
the work is started from scratch, therefore 
the random variable $\boldsymbol T$ measuring the time it takes to mine a block is memoryless, which means that for 
any $t_1, t_2 >0$
$$
\PP [\boldsymbol T > t_1+t_2 | T > t_2] = \PP [\boldsymbol T > t_1] \ .
$$
Therefore we have 
\begin{equation*}
\PP [\boldsymbol T > t_1+t_2]=\PP [\boldsymbol T > t_1+t_2 | T > t_2] . \PP [\boldsymbol T > t_2]=
\PP [\boldsymbol T > t_1] . \PP [\boldsymbol T > t_2] \ .
\end{equation*}
This equation and a continuity argument determines the exponential function and implies that $\boldsymbol T$ 
is an exponentially distributed random variable:
$$
f_{\boldsymbol T} (t) =\a e^{-\a t} \,
$$
for some parameter $\a >0$, the mining speed, with $t_0=1/\a = \EE [ \boldsymbol T]$. 

\medskip

If $(\bdT_1, \ldots , \bdT_n)$ is a sequence of independent identically distributed exponential random variables 
(for example $\bdT_k$ is the mining time of the $k$-th block), then the sum
$$
\bdS_n=\bdT_1+\ldots +\bdT_n
$$
is a random variable following a gamma density with parameters $(n,\alpha )$ (obtained by convolution of the 
exponential density):
$$
f_{\bdS_n} (t)= \frac{\a^n}{(n-1)!} t^{n-1}e^{-\a t} \ ,
$$
and cumulative distribution
$$
F_{\bdS_n}(t)=\int_0^t f_{\bdS_n}(u) du= 1-e^{-\a t} \sum_{k=0}^{n-1} \frac{(\a t)^k}{k!} \ .
$$

\medskip

We define the random process $\bdN(t)$ as the number of mined blocks at time $t$. Setting $S_0=0$, we have 
$$
\bdN(t)=\# \{ k\geq 1; \bdS_k \leq t\} =\max\{n\geq 0; \bdS_n <t\} \ .
$$

Since $\bdN(t)=n$ is equivalent to $\bdS_n\leq  t$ and $\bdS_{n+1} >t$ we get  
$$
\PP[\bdN(t)=n]=F_{\bdS_n}(t)-F_{\bdS_{n+1}}(t)=\frac{(\a t)^n}{n!}\ e^{-\a t} \  ,
$$
which means that $\bdN(t)$ has a Poisson distribution with expectation $\a t$.

\section{Mining race.}

We consider the situation described in section 11 of \cite {N} where a group of attacker miners 
attempts a double spend attack. The attacker group has a fraction $0<q<1/2$
of the total hash rate, and the rest, the honest miners, has a fraction $p=1-q$. Thus 
the probability that the attackers find the next block is $q$ while the probability for the honest miners is $p$.
Nakamoto computes the probability for the attackers to catch up when $z$ blocks have been mined by 
the honest group. In general to replace the chain mined by the honest 
miners and succeed a double spend the attackers need to mine $z+1$ blocks, i.e. to mine a longer chain. 
In the analysis it is assumed that we are not near an update of the difficulty which remains constant
\footnote{The difficulty is adjusted every 2016 blocks.}.

\medskip

The first discussion in section 11 of \cite {N} is about computing the probability $q_z$ of the attacker catching up
when they lag by $z$ blocks behind the honest miners. The analysis is correct and is similar to the 
Gamblers Ruin problem. We review this.

\medskip

\begin{lemma}
 Let $q_n$ be the probability of the event $E_n$, ``catching up from $n$ blocks behind''. We have
 $$
 q_n =(q/p)^n \ .
 $$
\end{lemma}

\begin{proof}

Note that after one more block has been mined, we have for $n\geq 1$,
$$
q_n= q q_{n-1}+p q_{n+1} \ ,
$$
and the only solution to this recurrence with $q_0=1$ and $q_n \to 0$ is $q_n =(q/p)^n$ (see \cite{F}).

%
%
%
%
\end{proof}

\medskip

We consider the random variables $\bdT$ and $\bdS_n$, resp. $\bdT'$ and $\bdS'_n$, associated 
to the group of honest, resp. attacker, miners.
And also consider the random Poisson process $\bdN (t)$, resp. $\bdN'(t)$. The random variables 
$\bdT$ and $\bdT'$ are clearly 
independent and have exponential distributions with parameters $\a$ and $\a'$. We have
$$
\PP[ \bdT' < \bdT] =\frac{\a'}{\a+\a'} \ ,
$$
so
\begin{align*}
 p&=\frac{\a}{\a+\a'} \ ,\\
 q&=\frac{\a'}{\a+\a'}\ .
\end{align*}
 
\medskip

Moreover, $\inf (\bdT, \bdT')$ is an exponentially distributed random variable with parameters $\alpha + \alpha'$
which represents the mining speed of the entire network, honest and attacker miners together. The Bitcoin 
protocol is calibrated such that $\alpha+\alpha'=\tau_0$ with $\tau_0 = 10 \, {\hbox{\rm min}}$. So we 
have 
\begin{align*}
 \EE[ \bdT]&=\frac1\alpha=\frac{\tau_0}{p} \ ,\\
 \EE[ \bdT']&=\frac{1}{\alpha'}=\frac{\tau_0}{q} \ .
\end{align*}

These results can also be obtained in the following way.
The hash function used in bitcoin block validation is $h(x)=\hbox{\rm {SHA256}}(\hbox{\rm {SHA256}}(x))$. 
The hashrate is the number of hashes per second 
performed by the miners. At a stable hashrate regime, the average time it takes to validate a block by the network 
is $\tau_0=10 \ {\hbox{\rm min}}$. If the difficulty is set to be $d \in (0,2^{256}-1]$, we validate a block 
when $h(BH)<d$, where $BH$ is the block header. 
The pseudo-random output of $\hbox{\rm {SHA256}}$ shows that we need to compute an average 
number of $m=2^{256}/d$ hashes to find a solution. Let 
$h$, resp. $h'$, be the hashrates of the honest miners, resp. the attackers. 
The total hashrate of the network is $h+h'$, and 
we have
\begin{align*}
p &=\frac{h}{h+h'} \ ,\\
q &=\frac{h'}{h+h'} \ .
\end{align*}

Let $t_0$, resp. $t'_0$, be the average time it 
takes to validate a block by the honest miners, resp. the attackers. We have 
\begin{align*} 
(h+h') \, \tau_0 &= m \ , \\
h \, t_0 &= m \ , \\
h' \, t'_0 &= m \ ,
\end{align*}
and from this we get that $\tau_0$ is half the harmonic mean of $t_0$ and $t'_0$,
$$
\tau_0=\frac{t_0 t'_0}{t_0+t'_0} \ ,
$$
and also
\begin{align*} 
 p &=\frac{t'_0}{t_0+t'_0}=\frac{\tau_0}{t_0}  \ ,\\
 q &=\frac{t_0}{t_0+t'_0}=\frac{\tau_0}{t'_0}  \ .
\end{align*}
Going back to the Poisson distribution parameters, we have
\begin{align*} 
 \a &=\frac{1}{t_0}=\frac{p}{\tau_0}  \ ,\\
 \a' &=\frac{1}{t'_0}=\frac{q}{\tau_0}  \ ,
\end{align*}
and we recover the relations
\begin{align*}
 p&=\frac{\a}{\a+\a'} \ ,\\
 q&=\frac{\a'}{\a+\a'}\ .
\end{align*}
 
\newpage 

\section{Nakamoto's analysis.}

Once the honest miners mine the $z$-th block, the attackers have mined $k$ blocks with 
a probability computed in the next section (Proposition \ref{prop_proba}). If $k>z$, then the attackers chain is adopted and the attack succeeds. 
Otherwise the probability they catch up is $(q/p)^z$ as computed above, therefore the probability $P$ of success 
of the attack is
$$
P=\PP[\bdN'(\bdS_z) \geq z]+\sum_{k=0}^{z-1} \PP[\bdN'(\bdS_z)=k] . q_{z-k} \ .
$$
Then Nakamoto makes the simplifying assumption that the blocks have been 
mined according to average expected time per block. This is asymptotically true when $z\to +\infty$ but false 
otherwise. More precisely, he approximates $\bdN'(\bdS_z)$ by $\bdN'(t_z)$ where 
$$
t_z=\EE[\bdS_z]=z\EE[\bdT] = \frac{z\tau_0}{p} \ .
$$
As we have seen above, the random variable $\bdN'(t_z)$ follows a Poisson distribution with parameter 
$$
\lambda = \alpha' t_z =\frac{z \alpha' \tau_0}{p}=\frac{zq}{p} \ .
$$
The final calculus in \cite {N} is then
\begin{align*}
 P_{SN}(z) &= \PP[\bdN'(t_z) \geq z]+\sum_{k=0}^{z-1} \PP[\bdN'(t_z)=k] . q_{z-k} \\
           &= 1-\sum_{k=0}^{z-1} \PP[\bdN'(t_z)=k] + \sum_{k=0}^{z-1}  \PP[\bdN'(t_z)=k] . q_{z-k} \\
           &= 1-\sum_{k=0}^{z-1} e^{-\lambda} \frac{\lambda^k}{k!} (1-q_{z-k}) \ .
\end{align*}

However, this analysis is not correct since $\bdN'(\bdS_z)\not= \bdN'(t_z)$.

\section{The correct analysis.}

Let $\bdX_n=\bdN'(\bdS_n)$ be the number of blocks mined by the attackers when the honest 
miners have just mined the $n$-th block. We compute the distribution for $\bdX_n$. 

\begin{proposition}\label{prop_proba}
 The random variable $\bdX_n$ has a negative binomial distribution with parameters $(n, p)$, i.e. for $k\geq 0$,
 $$
 \PP[\bdX_n=k] = p^n q^k  \binom{k+n-1}{k} \ .
 $$
\end{proposition}

\begin{proof}
  Let $k \geq 0$. We have that  $\bdN'$ and $\bdS_n$ are independent, therefore
  \begin{align*}
    \PP [ \bdX_n= k] & =  \int_0^{+ \infty} \PP [
    \bdN' ( \bdS_n) = k|\bdS_n \in [ t, t + dt]] \cdot \PP [
    \bdS_n \in [ t, t + dt]]\\
    & =  \int_0^{+ \infty} \PP [ \bdN' ( t)  = k] \cdot
    f_{\bdS_n}  ( t) dt\\
    & =  \int_0^{+ \infty}  \frac{( \a' t)^k}{k!} e^{- \a' t} \cdot
    \frac{\a^n}{( n - 1) !} t^{n - 1} e^{- \a t} dt\\
    & =  \frac{p^n q^k}{( n - 1) !k!} \cdot \int_0^{+ \infty} t^{k + n - 1}
    e^{- t} dt\\
    & =  \frac{p^n q^k}{( n - 1) !k!} \cdot ( k + n - 1) !
  \end{align*}
\end{proof}

Thus, contradicting to Nakamoto's claim, we have proved that the distribution of $\bdX_n$ is not a Poisson law with parameter $nq/p$.
Rosenfeld \cite{R} noticed the inacuracy of Nakamoto's approximation and proposed empirically the negative binomial distribution 
as a better approximation, not realizing that  this was the exact distribution\footnote{From \cite{R}. ``We will not use this assumption (Nakamoto's one), but 
rather model $m$ more accurately as a negative binomial variable'', and from \cite{R2}
``Instead of a Poisson distribution, I used a more accurate negative binomial distribution.''}. Only asymptotically we 
have convergence to the Poisson distribution:

\begin{proposition}
In the limit $n\to +\infty$, $q\to 0$, and $l_n = nq/p\to \lambda$ we have:
$$
\PP [ X_n = k] \to \frac{\lambda^k}{k!} e^{-\lambda} \ .
$$ 
\end{proposition}

\begin{proof} We have
\begin{align*}
    \PP [ X_n = k] & = \frac{n^n}{( n + l_n)^n}  \frac{l_n^k}{( n + l_n)^k} 
    \frac{( k + n - 1) !}{( n - 1) !k!}\\
    & = \frac{l_n^k}{k!}  \frac{1}{\left( 1 + \frac{l_n}{n} \right)^n}  \frac{n
    ( n + 1) \ldots ( n + k - 1)}{( n + l_n)^k}
\end{align*}
and the result follows using $\left( 1 + \frac{l_n}{n} \right)^n \rightarrow e^\lambda$.
\end{proof}

We can now compute the probability of success of the attackers catching up a longer chain. This computation was previously done in \cite{R}.
  
\begin{proposition} \textbf{(Probability of success of the attackers)}
 The probability of success by the attackers after $z$ blocks have been mined by the honest miners is
 $$
 P(z)=1-\sum_{k=0}^{z-1} \left (p^z q^k-q^z p^k\right ) \binom{k+z-1}{k}  \ .
 $$
\end{proposition}
  
\begin{proof}
As explained before, we have
\begin{align*}
P(z) &= \sum_{k>z} p^z q^k \binom{k+z-1}{k} + \sum_{k=0}^z \left (\frac{q}{p}\right )^{z-k} p^z q^k \binom{k+z-1}{k} \\
          &=1-\sum_{k=0}^z \left (p^z q^k-q^z p^k\right ) \binom{k+z-1}{k}\\
          &=1-\sum_{k=0}^{z-1} \left (p^z q^k-q^z p^k\right ) \binom{k+z-1}{k}
\end{align*}
\end{proof}

\begin{center}
\includegraphics[scale=0.45]{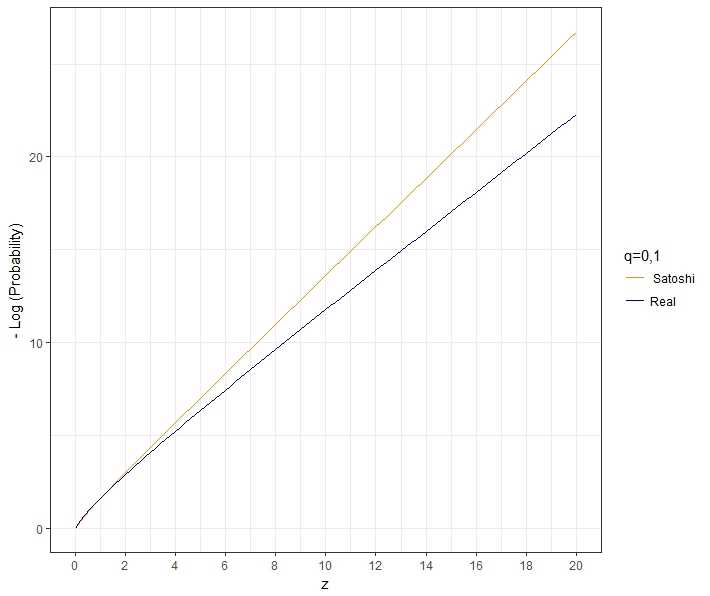} \\
{\footnotesize Figure 2. Nakamoto's and real probability}
\end{center}

\bigskip

\textbf{Numerical application.}

\bigskip

Converting to R code, given $0<q<1/2$ and $z\geq 0$, this simple 
function computes our probability $P(z)$:

\bigskip

{\tt

prob<-function(z,q)\{ 

p=1-q;

sum=1; 

for (k in 0:(z-1)) \{sum=sum-(p\^{}z*q\^{}k-q\^{}z*p\^{}k)*choose(k+z-1,k)\} ; 
 
return(sum) 

\}
}

\bigskip

We can compare with the probability $P_{SN}$ computed in \cite {N}.

\bigskip

%

\centerline{
\begin{tabular}{|l|l|l|}
\hline
$z$ &$P(z)$ & $P_{SN}(z)$ \\
\hline \hline
$0$ & $1.0000000$ & $1.0000000$ \\
$1$ & $0.2000000$ & $0.2045873$ \\
$2$ & $0.0560000$ & $0.0509779$ \\
$3$ & $0.0171200$ & $0.0131722$ \\
$4$ & $0.0054560$ & $0.0034552$ \\
$5$ & $0.0017818$ & $0.0009137$ \\
$6$ & $0.0005914$ & $0.0002428$ \\
$7$ & $0.0001986$ & $0.0000647$ \\
$8$ & $0.0000673$ & $0.0000173$ \\
$9$ & $0.0000229$ & $0.0000046$ \\
$10$ & $0.0000079$ & $0.0000012$ \\
\hline
\end{tabular}
}
\medskip

\centerline{\footnotesize{\textbf{Table 2. Probabilities for $q=0.1$.}}}

\bigskip
%
%

\centerline{
\begin{tabular}{|l|l|l|}
\hline
$z$ &$P(z)$ & $P_{SN}(z)$ \\
\hline \hline
$0$ & $1.0000000$ & $1.0000000$ \\
$5$ & $0.1976173$ & $0.1773523$ \\
$10$ & $0.0651067$ & $0.0416605$ \\
$15$ & $0.0233077$ & $0.0101008$ \\
$20$ & $0.0086739$ & $0.0024804$ \\
$25$ & $0.0033027$ & $0.0006132$ \\
$30$ & $0.0012769$ & $0.0001522$ \\
$35$ & $0.0004991$ & $0.0000379$ \\
$40$ & $0.0001967$ & $0.0000095$ \\
$45$ & $0.0000780$ & $0.0000024$ \\
$50$ & $0.0000311$ & $0.0000006$ \\
\hline
\end{tabular}
}

\medskip

\centerline{\footnotesize{\textbf{Table 3. Probabilities for $q=0.3$.}}}


$$
\begin{array}{|c||c|c|c|c|c|c|c|c|}
 \hline
 q & 0.10 & 0.15 & 0.20 & 0.25 & 0.30 & 0.35 & 0.40 & 0.45 \\ 
 \hline
 z & 6 & 9 & 13 & 20 & 32 & 58 & 133 & 539\\ 
 \hline
 z_{SN} & 5 & 8 & 11 & 15 & 24 & 41 & 81 & 340 \\ 
 \hline
\end{array}
$$
\medskip

\centerline{\footnotesize{\textbf{Table 4. Solving for $P$ less than 0.1\%.}}}

Therefore the correct results for bitcoin security are worse than those given in \cite {N}.
The explanation is that Nakamoto's result is correct only if the mining time by the honest 
miners is exactly the expected time. Times longer than average help the 
attackers.

\section{Closed-form formula.}

We give a closed-form formula for $P(z)$ using the regularized incomplete beta 
function $I_x(a,b)$ (see \cite {AS} (6.6.2)).

\begin{theorem}\label{prop_close}
 We have, with $s=4pq$, 
$$
P(z)= I_s(z,1/2) \ .
$$
\end{theorem}

We recall that the incomplete beta function is defined (see \cite {AS} (6.6.1)), for $a,b>0$ and $0\leq x\leq 1$, by 
$$
B_x(a,b)=\int_0^x t^{a-1}(1-t)^{b-1} \ dt \ , 
$$
and the classical beta function is defined (see \cite {AS} (6.2.1)) by $B(a,b)=B_1(a,b)$.

The Regularized Incomplete Beta Function is defined (see \cite {AS} (6.6.2) and (26.5.1)) by 
$$
I_x(a,b)=\frac{B_x(a.b)}{B(a,b)}= \frac{\Gamma(a+b)}{\Gamma(a) \Gamma(b)} B_x(a,b)\ .
$$
 
\begin{proof}
The cumulative distribution of a random variable $\bdX$ with negative binomial distribution, with 
$0<p<1$ and $q=1-p$ as usual (see \cite {AS} (26.5.26))) is given by
$$
F_\bdX (k)=\PP[\bdX\leq k]=\sum_{l=0}^k p^z q^l \binom{l+z-1}{l} =1-I_p(k+1,z)\ .
$$
This results from the formula (see \cite {AS} (6.6.1))
$$
I_p(k+1,z)=I_p(k,z) -\frac{p^k q^z}{kB(k,z)} \ ,
$$
that we prove by integrating by parts the definition of $B_x(a,b)$. 

Thus we get 
$$
P(z)= 1-I_p(z,z)+I_q(z,z) \ .
$$
Making the change of variables $t\mapsto 1-t$ in the integral definition, we also have 
a symmetry relation (see \cite {AS} (6.6.3))
$$
I_p(a,b)+I_q(b,a)=1 \ .
$$
Therefore we have $I_p(z,z)+I_q(z,z)=1$, and $P(z)= 2 I_q(z,z)$.
The result follows using (see \cite {AS} (26.5.14)),  $I_q(z,z)=\frac12 I_s(z, 1/2)$, where $s=4pq$.
 
\end{proof}

\section{Asymptotic and exponential decay.}

Nakamoto makes the observation (\cite {N} p.8), without proof, that the probability decreases 
exponentially to $0$ when $z\to +\infty$. We prove this fact for the true probability $P(z)$ 
using the closed-form formula from Proposition \ref{prop_close},

\begin{proposition}\label{prop_asymp}
When $z\to +\infty$ we have, with $s=4pq<1$,
$$
P(z) \sim \frac{s^{z}}{\sqrt{\pi (1-s) z}} \ .
$$
\end{proposition}

By integration by parts we get the following elementary version of Watson's Lemma:

\begin{lemma}\label{lemma_simpleWatson}
 Let $f\in C^1(\RR_+)$ with $f(0)\not= 0$ and absolutely convergent integral 
 $$
 \int_0^{+\infty} f(u)e^{-zu} \, du <+\infty \ ,
 $$
 then, when $z\to +\infty$, we have
 $$
 \int_0^{+\infty} f(u)e^{-zu} \, du \sim \frac{f(0)}{z} \ .
 $$
\end{lemma}

Then we get the following asymptotics (see also \cite{LS}):

\begin{lemma}
For $s,b\in \RR$, we have when $z\to +\infty$,
$$
B_s(z,b) \sim  \frac{s^z}{z} (1-s)^{b-1} \ .
$$
\end{lemma}

\begin{proof}
 Making the change of variable $u=\log (s/t)$ in the definition 
 $$
 B_s(z,b) = \int_0^s t^{z-1} (1-t)^{b-1} \, dt \ ,
 $$
 we get
 $$
 B_s(z,b) = s^z \int_0^{+\infty} (1-se^{-u})^{b-1} e^{-zu} \, du \ ,
 $$
 and the result follows applying Lemma \ref{lemma_simpleWatson} with $f(u)=(1-se^{-u})^{b-1}$.
\end{proof}
Now we end the proof of Proposition \ref{prop_asymp}. By Stirling asymptotics,
$$
B(z, 1/2)=\frac{\Gamma(z) \Gamma(1/2)}{\Gamma(z+1/2)}\sim \sqrt{\frac{\pi}{z}} \ ,
$$
so
$$
 I_s(z,1/2)= \frac{B_s(z,1/2)}{B(z,1/2)} \sim \frac{(1-s)^{-1/2} \frac{s^z}{z}}{\sqrt{\frac{\pi}{z}}} \sim \frac{s^{z}}{\sqrt{\pi (1-s) z}} \ .
$$

\section{A finer risk analysis.}

In practice, in order to avoid a double spend attack, the recipient of the bitcoin transaction waits for  $z\geq 1$ 
confirmations. But he also has the information on the time $\tau_1$ it took to confirm the transaction $z$ times. 
Obviously the probability of success of the attackers increases with $\tau_1$. The relevant parameter is the 
relative deviation from the expected time
$$
\kappa=\frac{\tau_1}{zt_0}=\frac{p \tau_1}{z\tau_0} \ .
$$

Our purpose is to compute the probability $P(z,\kappa )$ of success of the attackers. Note that $P(z,1)$ is the 
probability computed by Nakamoto \cite {N},
$$
P_{SN}(z)=P(z,1) \ .
$$

\medskip
\textbf{Computation of $P(z,\kappa)$.}
\medskip

The attackers mined $k\geq 0$ blocks during the time $\tau_1$ with probability that follows a Poisson 
distribution with parameter
$$
\lambda(z, \kappa) = \alpha' \tau_1=\kappa \frac{zq}{p} \ ,
$$
that means
$$
\PP[\bdN'(\tau_1)=k]=\frac{\left (\frac{zq}{p} \kappa\right)^k}{k!}\, e^{-\frac{zq}{p} \kappa} \ ,
$$
For $\kappa =1$ we recover Nakamoto's approximation.

\medskip

The cumulative Poisson distribution can be computed with the incomplete regularized gamma function 
(\cite {AS} (26.4))
$$
Q(s,x)=\frac{\Gamma(s,x)}{\Gamma(x)} \ ,
$$
where 
$$
\Gamma(s,x)=\int_x^{+\infty} t^{s-1}e^{-t} \, dt 
$$
is the incomplete gamma function and $\Gamma(s)=\Gamma(s,0)$ is the regular gamma function.
We have
$$
Q(z,\lambda)=\sum_{k=0}^{z-1}\frac{\lambda^k}{k!} e^{-\lambda} \ .
$$
We compute as before 
\begin{align*} 
 P(z, \kappa) &= \sum_{k=z}^{+\infty}  \frac{\left (\lambda(z, \kappa)
\right)^k}{k!}\, e^{-\lambda(z, \kappa)} + \sum_{k=0}^{z-1} \left (\frac{q}{p}\right )^{z-k} \, \frac{\left (\lambda(z, \kappa)
\right)^k}{k!}\, e^{-\lambda(z, \kappa)}  \\
 &=1-\sum_{k=0}^{z-1}\left ( 1 -\left (\frac{q}{p}\right )^{z-k}\right ) \, \frac{\left (\lambda(z, \kappa)
\right)^k}{k!}\, e^{-\lambda(z, \kappa)}  \\
 &= 1-Q(z, \kappa z q/p) + \left (\frac{q}{p}\right )^{z} e^{\kappa z\frac{p-q}{p}} Q(z, \kappa z)  \ .\\
\end{align*}

\begin{center}
\includegraphics[scale=0.5]{FirstGraphWithGreekLetterBis} \\
{\footnotesize Figure 3. Probability of success as a function of $\kappa$}
\end{center}

\bigskip

Thus we get a explicit closed-form formula for $P(z, \kappa)$,

\begin{theorem*}\label{thm_closedform2}
We have
$$
 P(z, \kappa)=1-Q(z, \kappa z q/p) + \left (\frac{q}{p}\right )^{z} e^{\kappa z\frac{p-q}{p}} Q(z, \kappa z) \ ,
$$
and 
$$
P_{SN}(z)=P(z, 1)=1-Q(z, z q/p) + \left (\frac{q}{p}\right )^{z} e^{ z\frac{p-q}{p}} Q(z, z) \ .
$$
\end{theorem*}

\section{Asymptotics of $P(z,\kappa)$ and $P_{SN}(z)$.}

We find the asymptotics of $Q(z, \lambda z)$ when $z\to +\infty$ for different values of $\lambda >0$.

\begin{lemma} \label{lemma_asymp}
We have
 \begin{enumerate}
  \item For $0<\lambda <1$, $Q(z, \lambda z) \to 1$ and $1-Q(z, \lambda z)\sim \frac{1}{1-\lambda}\frac{1}{\sqrt{2\pi z}}
  e^{-z(\lambda -1-\log \lambda)}$.
  \item For $\lambda=1$, $Q(z, z) \to 1/2$ and $1/2-Q(z, z)\sim \frac{1}{3\sqrt{2\pi z}}$.
  \item For $\lambda>1$, $Q(z,\lambda z )\sim \frac{1}{\lambda -1}\frac{1}{\sqrt{2\pi z}}
  e^{-z(\lambda -1-\log \lambda)}$.
 \end{enumerate}
\end{lemma}

\begin{proof}
(1) By \cite {DLMF} (8.11.6) and Stirling formula,for $\lambda <1$ we have
\begin{align*}
 1-Q(z,\lambda z) &= \frac{\gamma(z,\lambda z)}{\Gamma (z)}\\
&\sim \frac{z^z\lambda^ze^{-z\lambda}}{z! (1-\lambda)}\\
&\sim \frac{1}{1-\lambda} \frac{1}{\sqrt{2\pi z}} e^{-z(\lambda-1-\log \lambda )}
\end{align*}

(2) Also by \cite {DLMF} (8.11.12) and Stirling formula,
\begin{align*}
Q(z,z) &= \frac{z^{z-1} e^{-z} \sqrt{\frac{\pi z}{2}}}{(z-1)!}\\
&\sim \frac12 \frac{(z/e)^z \sqrt{2\pi z}}{z!}\\
&\to \frac12
\end{align*}
 and
 \begin{align*}
 \frac12-Q(z,z) &= \frac12-\frac{z^{z-1} e^{-z} \sqrt{\frac{\pi z}{2}} \left ( 1-\frac13 \sqrt{\frac{2}{\pi z}} +o(z^{-1/2})\right )}{(z-1)!}\\
&=\frac12-\frac12 \frac{\sqrt{2\pi z}(z/e)^z}{z!}  \left (1-\frac13 \sqrt{\frac{2}{\pi z}} +o(z^{-1/2})\right )\\
&=\frac12-\frac12 \frac{\sqrt{2\pi z}(z/e)^z}{\sqrt{2\pi z}(z/e)^z (1+\frac{1}{12 z} +o(z^{-1})) }  \left (1-\frac13 \sqrt{\frac{2}{\pi z}} +o(z^{-1/2})\right )\\
&=\frac12-\frac12  \left (1+\frac{1}{12 z} +o(z^{-1})\right ) \cdot \left (1-\frac13 \sqrt{\frac{2}{\pi z}} +o(z^{-1/2})\right )\\
&=\frac{1}{3\sqrt{2\pi z}}+o(z^{-1/2})
\end{align*}

(3) By \cite {DLMF} (8.11.7) and Stirling formula,for $\lambda >1$ we have
\begin{align*}
Q(z,\lambda z) &= \frac{\Gamma(z,\lambda z)}{\Gamma (z)}\\
&\sim \frac{(\lambda z)^z e^{-z\lambda}}{z! (\lambda-1)}\\
&\sim \frac{1}{\lambda-1} \frac{1}{\sqrt{2\pi z}} e^{-z(\lambda-1-\log \lambda )}
\end{align*}

\end{proof}

For $x>0$ we define $c(x)=x-1-\log x$, which is positive since the graph of $x\mapsto 1-x$ is the tangent at $x=1$
to the concave graph of the logarithm function. We denote $0<\lambda=q/p<1$.

\medskip

We have that the Nakamoto probability $P_{SN}(z)$ also decreases exponentially with $z$ as claimed 
by Nakamoto in \cite {N} without proof.

\begin{proposition} \label{prop_exp}
We have for $z\to +\infty$,
$$
P_{SN}(z)\sim \frac{e^{-z c(\lambda)}}{2}
$$
\end{proposition}

\begin{proof}
The result follows from the closed-form formula from Theorem \ref{thm_closedform2}, 
$$
P(z,\kappa)=1-Q(z, \kappa zq/p)+(q/p)^z e^{\kappa z \frac{p-q}{q}} Q(z, \kappa z) \ ,
$$
and then from points (1) and (2) of Lemma \ref{lemma_asymp},
$$
1-Q\left(z, \frac{q}{p} z\right )  =o\left (e^{-zc(q/p)} \right )  \ ,
$$
and 
$$
\left (\frac{q}{p}\right )^{z} e^{ z (1-\frac{q}{p})} Q(z, z)\sim \frac{1}{2} e^{-zc(q/p)} \ .
$$
\end{proof}

More generally, we have five different regimes for the asymptotics of $P(z,\kappa)$ for $0<\kappa<1$, 
$\kappa =1$, $1<\kappa<p/q$, $\kappa = p/q$ and $\kappa >p/q$.

\begin{proposition} \label{prop_asymp2}
We have for $z\to +\infty$,
\begin{enumerate}
  \item For $0<\kappa<1$, 
  $$
  P(z,\kappa) \sim \frac{1}{1-\kappa \lambda} \frac{1}{\sqrt{2\pi z}}
  e^{-zc(\kappa \lambda)} \ .
  $$
  \item For $\kappa=1$, 
  $$
  P(z,1)=P_{SN}(z)\sim \frac{1}{2} e^{-zc(\lambda)} \ .
  $$
  \item For $1<\kappa<p/q$, 
  $$
  P(z,\kappa)\sim \frac{\kappa (1-\lambda)}{(\kappa -1)(1-\kappa \lambda)} 
  \frac{1}{\sqrt{2\pi z}} e^{-zc(\kappa \lambda)} \ .
  $$
  \item For $\kappa=p/q$, $P(z, p/q)\to 1/2$ and 
  $$
  P(z, p/q)-1/2 \sim \frac{1}{2\pi z} \left (\frac13+\frac{q}{p-q}\right ) \ .
  $$
  \item For $p/q<\kappa$, $P(z,\kappa) \to 1$ and
  $$
  1-P(z,\kappa)\sim \frac{\kappa (1-\lambda)}{(\kappa -1)(\kappa \lambda-1)} 
  \frac{1}{\sqrt{2\pi z}} e^{-zc(\kappa \lambda)} \ .
  $$
\end{enumerate}
\end{proposition}

\begin{proof}
(1) If $\kappa <1$ then also $\kappa q/p <1$, and
$$
1-Q(z,\kappa z q/p) \sim \frac{1}{1-\kappa q/p} \frac{1}{\sqrt{2\pi z}} e^{-z(\kappa q/p -1-\log(\kappa q/p))} \ ,
$$
and
\begin{align*}
 (q/p)^z e^{\kappa z \frac{p-q}{q}} &= e^{-z(\kappa q/p -1-\log(\kappa q/p))}  \\
&=e^{-z(1-\kappa)(1-q/p)} \cdot e^{-z(q/p -1-\log(q/p))} \ ,
\end{align*}
and then
\begin{align*}
 \frac{(q/p)^z e^{\kappa z \frac{p-q}{q}}}{1-Q(z,\kappa z q/p) } &\sim (1-\kappa q/p)\cdot \sqrt{2\pi z} \cdot e^{-z(1-\kappa)(1-q/p)} \cdot 
e^{-z(q/p -1-\log(q/p)-(\kappa q/p -1-\log(\kappa q/p)))} \\
&\sim (1-\kappa q/p)\cdot \sqrt{2\pi z} \cdot e^{-z(1-\kappa)(1-q/p)}\cdot e^{-z(1-\kappa)q/p} \cdot e^{-z\log \kappa} \\
&\sim (1-\kappa q/p)\cdot \sqrt{2\pi z} \cdot e^{-z(1-\kappa-\log \kappa)} =o(1) \ .
\end{align*}
Since $Q(z, \kappa z)\to 1$ we have,
\begin{align*}
P(z,\kappa)&= 1-Q(z, \kappa zq/p)+ (q/p)^z e^{\kappa z \frac{p-q}{q}} Q(z, \kappa z) \\
&\sim 1-Q(z, \kappa zq/p) \\
&\sim \frac{1}{(1-\kappa q/p) \sqrt{2\pi z}} \cdot e^{-z(\kappa q/p-1-\log (\kappa q/p))} \ .
\end{align*}

\medskip

(2) This was proved in Proposition \ref{prop_exp}.

\medskip

(3) When $1<\kappa < p/q$ then by Lemma \ref{lemma_asymp},
$$
(q/p)^z e^{\kappa z \frac{p-q}{q}} Q(z, \kappa z) \sim \frac{1}{(\kappa -1) \sqrt{2\pi z}} \cdot e^{-z(\kappa q/p-1-\log (\kappa q/p))} \ ,
$$
and
$$
1-Q(z, \kappa zq/p) \sim \frac{1}{(1-\kappa q/p) \sqrt{2\pi z}} \cdot e^{-z(\kappa q/p-1-\log (\kappa q/p))} \ .
$$
So we have
\begin{align*}
 P(z,\kappa) &\sim \left ( \frac{1}{1-\kappa q/p} +\frac{1}{\kappa-1}\right ) \cdot \frac{1}{\sqrt{2\pi z}} \cdot  e^{-z(\kappa q/p-1-\log (\kappa q/p))} \\
&\sim  \frac{\kappa (1-q/p)}{(\kappa-1)(1-\kappa q/p)}\, \frac{1}{\sqrt{2\pi z}} \cdot  e^{-z(\kappa q/p-1-\log (\kappa q/p))} \ .
\end{align*}

\medskip
(4) The previous asymptotic at the start of the proof of (3) is also valid for $1<\kappa=p/q$ and gives
$$
(q/p)^z e^{\kappa z \frac{p-q}{q}} Q(z, \kappa z) \sim \frac{q}{p-q}\, \frac{1}{\sqrt{2\pi z}}  \ ,
$$
and by Lemma \ref{lemma_asymp},
\begin{align*}
 P(z,p/q) &= 1-Q(z,z) + (q/p)^z e^{\kappa z \frac{p-q}{q}} Q(z, \kappa z) \\
&=\frac12 + \frac{1}{\sqrt{2\pi z}}\, \left ( \frac13 +\frac{q}{p-q}\right ) +o(1/\sqrt{z}) \ .
\end{align*}

\medskip

(5) For $\kappa > p/q$ we use again the same asymptotic of (3) to get 
$$
Q(z, \kappa z q/p) \sim \frac{1}{\kappa q/p -1}\, \frac{1}{\sqrt{2\pi z}} \, e^{-z(\kappa q/p-1-\log (\kappa q/p))} \ ,
$$
and again
$$
(q/p)^z e^{\kappa z \frac{p-q}{q}} Q(z, \kappa z) \sim \frac{1}{(\kappa -1) \sqrt{2\pi z}} \, e^{-z(\kappa q/p-1-\log (\kappa q/p))} \ ,
$$
so
\begin{align*}
1-P(z,\kappa) &\sim \left ( \frac{1}{\kappa q/p -1} - \frac{1}{\kappa -1} \right ) \sqrt{2\pi z} \, e^{-z(\kappa q/p-1-\log (\kappa q/p))} \\
&\sim \frac{\kappa (1-q/p)}{(\kappa q/p-1)(\kappa -1)} \sqrt{2\pi z} \, e^{-z(\kappa q/p-1-\log (\kappa q/p))} \ .
\end{align*}

\end{proof}

\section{Comparing asymptotics of $P(z)$ and $P_{SN}(z)$.}

We have an asymptotic comparison,
\begin{proposition}
 We have for $z\to +\infty$,
$$
P_{SN}(z) \prec P(z) \ . 
$$
\end{proposition}
\begin{proof}
Note that
$$
 \frac{q}{p} - 1 - \log \left(  \frac{q}{p} \right) - \log
    \left(\frac{1}{4pq}\right ) =  2 \left[ \frac{1}{2 p} - 1 - \log \left( \frac{1}{2 p} \right)
    \right] >0
$$

So with $s=4pq<1$ we have
$$
0<\log \frac{1}{s}<\frac{q}{p}-1-\log \frac{q}{p}=c(q/p)=c(\lambda) \ ,
$$
and for $z$ large
$$
P_{SN}(z) <e^{-z c(\lambda)} \prec \frac{s^{z}}{\sqrt{\pi (1-s) z}} \sim P(z) \ .
$$
\end{proof}

As we will see later we can be more explicit about the inequality between $P_{SN}(z)$ and $P(z)$.

\section{Recovering $P(z)$ from $P(z,\kappa)$.}

We have seen above that $P_{SN}(z)$ can be recover from $P(z, \kappa)$ by taking the value at $\kappa=1$.
It turns out that we can also recover $P(z)$ as a weighted average on $\kappa$ of $P(z, \kappa)$.

\begin{theorem*}
We have 
$$
P(z)=\int_0^{+\infty } P(z, \kappa) \, d\rho_z(\kappa)
$$
with the density function
$$
d\rho_z(\kappa) = \frac{z^z}{(z-1)!}\kappa^{z-1} e^{-z\kappa}\, d\kappa \ .
$$ 
\end{theorem*}
We check that 
$$
\int_0^{+\infty } d\rho_z(\kappa )=1 \ .
$$
We can write 
$$
P(z)=1-\sum_{k=0}^{z-1} f_k(\kappa) \ ,
$$
where
$$
f_k(\kappa )=\left (1-\left ( \frac{q}{p}\right )^{z-k}\right ) \frac{(zq/p)^k}{k!} \kappa^k e^{\frac{zq}{p}\kappa} \ .
$$
Then the Theorem follows from a direct computation,
\begin{lemma}
For $k\geq 0$, we have
$$
\int_0^{+\infty} f_k(\kappa) \, d\rho_z(\kappa) =(p^zq^k-q^zp^k) \binom{k+z-1}{k} \ .
$$
\end{lemma}

We give a second more conceptual proof.

\begin{proof}
Consider the random variable
$$
\bdkappa=\frac{p}{z\tau_0} \bdS_z \ .
$$
We have seen above that $\bdS_z \sim \Gamma (z , \alpha )$ so 
$\bdkappa \sim \Gamma(z, \alpha \frac{z\tau_0}{p})=\Gamma (z, z)$. So the density $d\rho_z$ is the 
distribution of $\bdkappa$. It is enough to prove that
$$
P(z)=\EE \left[P(z,\bdkappa )\right] \ .
$$
We have 
\begin{align*}
 P(z) &=\PP[\bdN'(\bdS_z)\geq z] + \sum_{k=0}^{z-1} \PP [\bdN'(\bdS_z)=k]\, .\, q_{z-k} \\
 &= 1- \sum_{k=0}^{z-1} (1-q_{z-k}) \PP[\bdN'(\bdS_z)=k] \ .
\end{align*}
And by conditioning by $\bdS_z$ we get
\begin{align*}
 P(z) &= 1- \sum_{k=0}^{z-1} (1-q_{z-k}) \EE[\PP[\bdN'(\bdS_z)=k|\bdS_z]] \\
 &=1- \EE \left [\sum_{k=0}^{z-1} \frac{(\alpha' \bdS_z )^k}{k!} e^{-\alpha' \bdS_z}\right ]+
 \left ( \frac{q}{p}\right )^z \, \EE\left[ e^{\alpha' \frac{p-q}{q} \bdS_z} \, \sum_{k=0}^{z-1}
 \frac{\left (\frac{\alpha' p}{q} \bdS_z \right )^k}{k!}   e^{-\frac{\alpha'p}{q} \bdS_z} \right ] \\
 &=\EE\left[ 1-Q\left (z, \frac{zq}{p} \bdkappa \right ) + \left ( \frac{q}{p}\right )^z 
 e^{z\left (1-\frac{q}{p}\right ) \bdkappa} Q(z, z\bdkappa ) \right ] \\
 &=\EE \left[P(z,\bdkappa )\right] \, ,
\end{align*}
since $\PP[\bdN'(\bdS_z)=k|\bdS_z]=\frac{(\alpha' \bdS_z)^k}{k!} e^{\alpha' \bdS_z}$, $q_{z-k}=(q/p)^{z-k}$,
and 
$$
Q(z,x)=\sum_{k=0}^{z-1} \frac{x^k}{k!}e^{-x} \ .
$$
\end{proof}
We also note that $\EE[\bdkappa]=1$.
\bigskip

\section{Range of $\kappa$.}

The probability to observe a deviation greater than $\kappa$ is $\PP
[\bdkappa > \kappa]$ with $\bdkappa= \frac{p}{z \tau_0}
\bdS_z$. We have that $\bdkappa$ follows a $\Gamma$-distribution, $\bdkappa \sim \Gamma (z, z)$,  so
\begin{align*}
  \PP [\bdkappa> \kappa] & =  \frac{1}{\Gamma (z)} 
  \int_{\kappa}^{+ \infty} z^z t^{z - 1} e^{- zt} \, dt\\
  & =  \frac{1}{\Gamma (z)}  \int_{\kappa z}^{+ \infty} t^{z - 1} e^{-
  t} \, dt\\
  & =  \frac{\Gamma (z, \kappa z)}{\Gamma (z)}\\
  & =  Q (z, \kappa z) \ .
\end{align*}
Then, by Lemma \ref{lemma_asymp}, $\PP [\bdkappa> \kappa] \sim
\frac{1}{\kappa - 1}  \frac{1}{\sqrt{2 \pi z}} e^{- zc (\kappa)}$ for
$\kappa > 1$. Note that this probability does not depend on $p$. For $z = 6$, we have
$\PP [\bdkappa> 4] \approx 3 \cdot 10^{- 6}$ and for $z = 10$,
$\PP [\bdkappa> 4] \approx 4 \cdot 10^{- 9}$. So, in practice, the probability to
have $\bdkappa> 4$ is very unlikely. Below, we have represented the
graph of $\kappa \longmapsto P (z, \kappa)$ for different values of $z$ ($q =
0.1$) and $0 < \kappa < 4$.

\begin{center}
\includegraphics[scale=0.5]{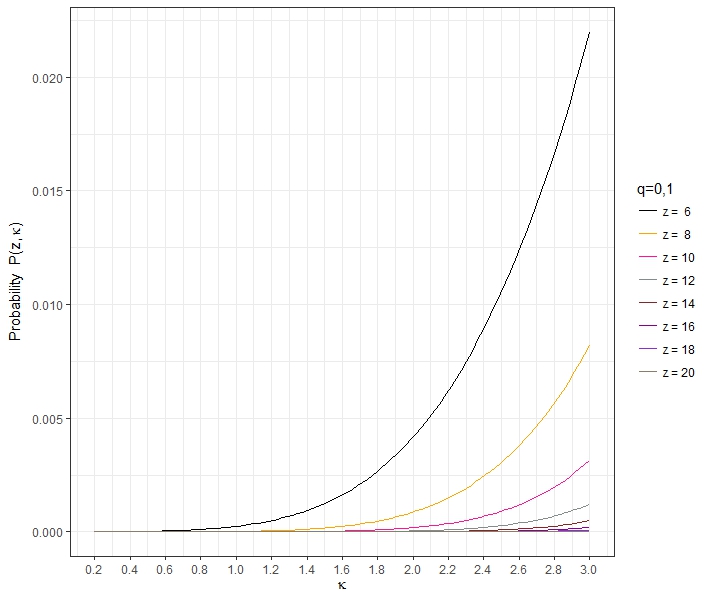} \\
{\footnotesize Figure 4. Probability $P(z,\kappa)$ as a function of $\kappa$}
\end{center}

We see that $\kappa \longmapsto P (z, \kappa)$ is convex in the range of values of $\kappa$ considered. We study 
the convexity in more detail in the next section.

\section{Comparing $P_{SN}(z)$ and $P(z)$.}

Now we study the convexity of $\kappa \mapsto P(z,\kappa )$.
Recall that $\lambda=q/p <1$. From Theorem \ref{thm_closedform2} we have
$$
P(z,\kappa)= 1 -Q(z,z\lambda \kappa)+\lambda^z e^{z(1-\lambda) \kappa} Q(z, z \kappa) \ .
$$
Since
$$
\Gamma (z) \partial_2 Q(z,x) = - x^{z-1} e^{-x} \ ,
$$
we get, after some cancellations,
$$
\Gamma (z) \, \partial_2 P(z,\kappa ) = \lambda^z z (1-\lambda)e^{z(1-\lambda) \kappa} \Gamma (z, z\kappa )\ .
$$
We observe that $\partial_2 P(z,\kappa ) >0$, so $P(z,\kappa )$ is an increasing function 
of $\kappa$ as expected. For the second derivative we have
\begin{align*}
\Gamma (z) \, \partial^2_2 P(z,\kappa ) &= \lambda^z z^2 (1-\lambda)e^{z(1-\lambda) \kappa} 
\left [(1-\lambda) \Gamma (z, z\kappa )- (z \kappa)^{z-1} e^{-\kappa z}\right ]\\
&=\lambda^z z (1-\lambda)e^{-\lambda \kappa z} (z \kappa)^{z} 
\left [(1-\lambda) Q (z, z\kappa ) z! e^{\kappa z}(z \kappa)^{-z}- \kappa^{-1} \right ] \ .
\end{align*}
Therefore we study the sign of
\begin{align*}
g_{\lambda, z} (\kappa) &=(1-\lambda) Q (z, z\kappa )z! e^{\kappa z}(z \kappa)^{-z}- \kappa^{-1}  \\
&=(1-\lambda) \sum_{k=0}^{z-1} \frac{z!}{z^{z-k} k!} \frac{1}{\kappa^{z-k}} -\kappa^{-1} \\
&= \frac{1-\lambda}{\kappa} \left ( \left(1-\frac{1}{z}\right ) \frac{1}{\kappa} +  \left(1-\frac{1}{z}\right )
\left(1-\frac{2}{z}\right ) \frac{1}{\kappa^2}+\ldots \right ) -\frac{\lambda}{\kappa}
\end{align*}

For $z=1$ we have 
$$
g_{\lambda, 1}(\kappa)=-\lambda/\kappa <0 \ ,
$$
therefore $\kappa \mapsto P(1, \kappa)$ is a concave function and by Jensen's inequality 
 $$
 P(1)=\int_0^{+\infty} P(z,\kappa ) \, d\rho_1(\tau) \leq P(1,\bar \kappa)=P(1,1)=P_{SN}(1)\ .
 $$
\begin{corollary}
 We have (for all $0<q<1/2$)
 $$
 P(1)\leq P_{SN}(1) \ .
 $$
\end{corollary}
In general, for $z\geq 2$, we have the reverse inequality. To determine the sign of $g_{\lambda, z}$ we 
study its zeros. 

The equation to solve is 
 \begin{equation*}
 \left (1-\frac{1}{z}\right ) \frac{1}{\kappa} + \left (1-\frac{1}{z}\right )
 \left (1-\frac{2}{z}\right ) \frac{1}{\kappa^2}+\ldots+\left (1-\frac{1}{z}\right )
 \ldots \left (1-\frac{z-1}{z}\right )  \frac{1}{\kappa^{z-1}} =\frac{\lambda}{1-\lambda} \ .
 \end{equation*}
This is a polynomial equation in $1/\kappa$, the coefficients are increasing on $z$, and the left hand side is 
decreasing on $\kappa \in (0,+\infty )$ from $+\infty$ to $0$, therefore there is a unique solution $\kappa (z)$,
and 
$$
\kappa (2) <\kappa (3) < \ldots
$$
We compute 
$$
\kappa (2)=\frac{1-\lambda}{2\lambda}=\frac{1}{2q}-1 >0 \ .
$$
In this case the function $\kappa\mapsto P(z,\kappa)$ is convex only in the interval $(0, \kappa(z))$. 
For $z$ large, most of the support of the measure $d\rho_z$ is contained in this interval and we have by
Jensen's inequality
 $$
 P(z) \approx \int_0^{\kappa (z)} P(z, \kappa) \, d\rho_z(\kappa) \geq P(z, \bar \kappa_z)
 \approx P(z,1)=P_{SN}(z) \ ,
 $$
 where 
$$
 \bar \kappa_z =\int_0^{\kappa (z)} \kappa \, d\rho_z(\kappa) \approx \int_0^{+\infty} \kappa \, d\rho_z(\kappa) =1\, .
$$

We can get some estimates on $\kappa(z)$ for $z\to +\infty$. The first observation is that 
for $z$ large we have $\kappa(z)>1$. The asymptotic limits for $Q(z,\kappa z)$ for $\kappa <1$ and $\kappa =1$ 
(Lemma \ref{lemma_asymp}) and Stirling
asymptotic formula give that 
$$
Q(z, \kappa z) z! e^{\kappa z} (z\kappa )^z \to +\infty \ ,
$$
and $g_{\lambda, z}(\kappa ) \not=0$.

For $\kappa >1$, we can use the asymptotic \cite {DLMF} (8.11.7),  $z\to +\infty$,
$$
\Gamma(z,\kappa z) \sim \frac{(\kappa z)^z e^{-\kappa z}}{(\kappa-1)z} 
$$
and 
$$
(1-\lambda) \Gamma (z, \kappa z)-(\kappa z)^{z-1} e^{-\kappa z} \sim 
(\kappa z)^{z-1} e^{-\kappa z} \left ( (1-\lambda) \frac{\kappa}{\kappa-1} -1\right ) \ ,
$$
thus, since
$$
g_{\lambda, z}(\kappa) = (1-\lambda) \Gamma (z, \kappa z) z e^{\kappa z}(\kappa z)^{-z}-\kappa^{-1} \ ,
$$
we have
$$
g_{\lambda, \infty}(\kappa)=\lim_{z\to +\infty} g_{\lambda, z}(\kappa) = \frac{1}{\kappa}\left ( (1-\lambda) \frac{\kappa}{\kappa-1} -1 \right ) = \frac{1-\lambda}{\kappa-1} -\frac{1}{\kappa}.
$$
Now, if 
$$
\kappa(\infty) =\lim_{z\to +\infty} \kappa (z) \ ,
$$
we have $g_{\lambda, \infty}(\kappa_\infty) = 0$, so we get:

\begin{proposition}
$$
\kappa(\infty) =\lim_{z\to +\infty} \kappa (z) =\lambda^{-1} =\frac{p}{q} \ .
$$ 
\end{proposition}

Using the second order asymptotic (\cite{DLMF} (8.11.7)), for $\kappa >1$, $z\to +\infty$,
$$
\Gamma(z,\kappa z) \sim \frac{(\kappa z)^z e^{-\kappa z}}{z(\kappa -1)} 
\left ( 1 -\frac{\kappa }{(\kappa -1)^2 z}\right ) \ ,
$$
so
$$
g_{\lambda , z}(\kappa)\sim \frac{1-\lambda}{\kappa -1} \left (1 -\frac{\kappa}{(\kappa-1)^2 z} \right ) -\kappa^{-1} \ .
$$
Writing 
$$
\kappa (z)=\frac{p}{q}-\frac{a}{z} +o(z^{-1}) \ ,
$$
and using 
$$
\frac{1-\lambda}{\kappa (z) -1} \left (1 -\frac{\kappa (z)}{(\kappa (z)-1)^2 z} \right ) -\kappa (z)^{-1}
$$
we get 
\begin{proposition}
For $z\to +\infty$
$$
\kappa (z)=\frac{p}{q}-\frac{p^2}{q (p-q)} \, \frac{1}{z} +o(z^{-1}) \ .
$$
\end{proposition}

Also we have 
$$
\frac{p}{q}-1 > \frac{p^2}{q (p-q)} \, \frac{1}{z}
$$
for 
$$
z > \left ( \frac{p}{p-q}\right )^2 \ ,
$$
so, for $z$ of the order of $( 1-\lambda)^{-2}$ we have $\kappa (z) >1$.

\pagebreak

\section{Bounds for $P (z)$}

Remember that we have set $s =4 pq$. We have the following inequality that is a 
particular case of more general Gautschi's inequalities \cite {G}:

\begin{lemma} \label{lemma_gautschi}
  Let $z \in \RR_+$. We have
$$
 \sqrt{\frac{z}{z + \frac{1}{2}}} \leq \frac{\Gamma \left( z +
     \frac{1}{2} \right)}{\sqrt{z} \, \Gamma ( z)} \leq 1  \ .
$$
\end{lemma}

\begin{proof}
  By Cauchy-Schwarz inequality, we have:
  \begin{align*}
    \Gamma \left( z + \frac{1}{2} \right) & =  \int_0^{+ \infty} t^{z -
    \frac{1}{2}} e^{-t} \, dt\\
    & \leq  \int_0^{+ \infty}  \left( t^{\frac{z}{2}} e^{-
    \frac{t}{2}} \right) \cdot \left( t^{\frac{z}{2} - \frac{1}{2}} e^{-
    \frac{t}{2}} \right) \, dt\\
    & \leq  \left( \int_0^{+ \infty} t^z e^{- t} \, dt
    \right)^{\frac{1}{2}} \cdot \left( \int_0^{+ \infty} t^{z - 1} e^{-
    t} \, dt \right)^{\frac{1}{2}}\\
    & \leq  \Gamma ( z + 1)^{\frac{1}{2}} \cdot \Gamma (
    z)^{\frac{1}{2}}\\
    & \leq  (z \Gamma ( z))^{\frac{1}{2}} \cdot \Gamma (
    z)^{\frac{1}{2}}\\
    & \leq  \sqrt{z} \Gamma (z)
  \end{align*}
  On the other side, the last inequality with $z$ replaced by $z +
  \frac{1}{2}$ gives:
  $$ 
z \Gamma ( z) = \Gamma \left( z + \frac{1}{2} + \frac{1}{2} \right)
     \leq \sqrt{z + \frac{1}{2}} \Gamma \left( z + \frac{1}{2} \right) 
$$
\end{proof}

\begin{lemma} \label{boundpnsz}
 For $z > 1$, we have
$$
  \sqrt{\frac{z}{z + \frac{1}{2}}} \cdot \frac{s^z}{\sqrt{\pi z}} \leq
     P(z) \leq \frac{1}{\sqrt{1-s}} \cdot \frac{s^z}{\sqrt{\pi z}}  \ .
$$
\end{lemma}

\begin{proof}
  The function $x \longmapsto ( 1 - x)^{- \frac{1}{2}}$ is non-decreasing. So,
  by definition of $I_s$ and the upper bound of the inequality of Lemma \ref{lemma_gautschi}, we have
  \begin{align*}
    P (z) = I_s \left( z, \frac{1}{2} \right) & =  \frac{\Gamma
    \left( z + \frac{1}{2} \right)}{\Gamma \left( \frac{1}{2} \right) \Gamma (
    z)}  \int_0^s t^{z - 1}  ( 1 - t)^{- \frac{1}{2}} \, dt\\
    & \leq  \frac{1}{\sqrt{\pi}}  \frac{\Gamma \left( z +
    \frac{1}{2} \right)}{\Gamma ( z)}  \int_0^s t^{z - 1}  ( 1 - s)^{-
    \frac{1}{2}} \, dt\\
    & \leq  \frac{\Gamma \left( z + \frac{1}{2}
    \right)}{\sqrt{z} \, \Gamma ( z)}  \cdot \frac{s^z}{\sqrt{\pi ( 1 - s)
    z}} \\
    & \leq \frac{1}{\sqrt{1-s}} \cdot \frac{s^z}{\sqrt{\pi z}} \ .
  \end{align*}
 In the same way, using the lower bound of the inequality of Lemma \ref{lemma_gautschi}, we have
  \begin{align*}
    P  ( z) = I_s \left( z, \frac{1}{2} \right) & \geq 
    \frac{1}{\sqrt{\pi}}  \frac{\Gamma \left( z + \frac{1}{2}
    \right)}{\Gamma ( z)}  \int_0^s t^{z - 1} \, dt\\
    & \geq  \frac{\Gamma \left ( z + \frac12 \right )}  {\sqrt{z} \, \Gamma (z)}   \cdot \frac{s^z}{\sqrt{\pi z}}\\
    & \geq  \sqrt{\frac{z}{z + \frac{1}{2}}} \cdot \frac{s^z}{\sqrt{\pi
    z}} \ .
  \end{align*}
\end{proof}

Note that this gives again the exponential decrease of Nakamoto's probability.

\section{An upper bound for $P_{SN} (z)$}

\begin{proposition} \label{inepnsz}
 We have,
  $$
    P_{SN} (z) < \frac{1}{1 - \frac{q}{p}}  \frac{1}{\sqrt{2
    \pi z}} e^{- \left( \frac{q}{p} - 1 - \log \frac{q}{p} \right) z}
    + \frac{1}{2} e^{- \left( \frac{q}{p} - 1 - \log \left(  \frac{q}{p}
    \right) z \right)}  
  $$
\end{proposition}
 This upper bound is quite sharp in view of the asymptotics in Proposition \ref{prop_asymp2} (2).

\begin{lemma} \label{inegb}
 Let $z \in \NN^{\ast}$ and $\lambda \in \RR_+^{\ast}$.
  \begin{enumerate}
    \item If $\lambda \in ]0,1[$, then $1 - Q ( z, \lambda z)
    < \frac{1}{1 - \lambda}  \frac{1}{\sqrt{2 \pi z}} e^{- (
    \lambda - 1 - \log \lambda) z}$
    
    \item $Q ( z, z) < \frac{1}{2}$
  \end{enumerate}
\end{lemma}

\begin{proof}
  For (1) We use \cite {DLMF} (8.7.1)
 $$
    \gamma ( a, x) = e^{- x} x^a  \sum_{n = 0}^{\infty}  \frac{\Gamma
    ( a)}{\Gamma ( a + n + 1)} x^n \ ,
 $$
 which is valid for $a, x \in \RR$. Let $\lambda \in] 0, 1 [$. Using $\Gamma ( z + 1) = z \Gamma ( z)$, we get:
  \begin{align*}
    \gamma (z, \lambda z) & =  e^{-\lambda z}  (\lambda z)^z  \sum_{n
    = 0}^{+ \infty}  \frac{\Gamma (z)}{\Gamma (z + n + 1)}  (\lambda z)^n \\
    & =  e^{-\lambda z}  (\lambda z)^z  \left( \frac{1}{z} + \frac{1}{z ( z + 1)}  ( \lambda z) +
    \frac{1}{z ( z + 1) ( z + 2)}  (\lambda z)^2 + \ldots \right)\\
    & \leq e^{-\lambda z}  (\lambda z)^z   \left( \frac{1}{z} + \frac{1}{z^2}  (
    \lambda z) + \frac{1}{z^3}  (\lambda z)^2 + \ldots \right) \\
    & \leq e^{-\lambda z}  (\lambda z)^z  \frac{1}{z}  \frac{1}{1 - \lambda}\\
    & \leq\frac{\lambda^z z^{z - 1} e^{-\lambda z}}{1 -
    \lambda}
  \end{align*}
  On the other hand, by \cite {DLMF} (5.6.1), we have
  $$
    \frac{1}{\Gamma ( z)} < \frac{e^z}{\sqrt{2 \pi z} z^{z - 1}} \ ,
  $$
  and for any $0 < \lambda < 1$,
   \begin{align*}
    1 - Q ( z, \lambda z) & =  \frac{\gamma ( z, \lambda z)}{\Gamma ( z)}\\
    & < \frac{1}{1 - \lambda}  \frac{1}{\sqrt{2 \pi z}} e^{- (
    \lambda - 1 - \log \lambda)z}
  \end{align*}
  
  For (2) this comes directly from \cite {DLMF} (8.10.13).
\end{proof}

Recalling that  $P_{SN} (z) = P (z, 1) = 1 - Q \left( z,  \frac{q}{p} z \right)
+ (q/p)^z e^{z (p-q)/p} Q ( z, z)$, we get Proposition \ref{inepnsz}.

\pagebreak

\section{Comparing again $P_{SN}(z)$ and $P(z)$.}

The aim of this section is to compute an explicit rank $z_0$ (no sharp) for which
$P_{SN}  ( z) < P  ( z)$ for $z \geq z_0$.

\begin{lemma} \label{m1}
 Let $\alpha > 0$. For all $x > \log \alpha$, $e^x - \alpha x > \frac{\alpha}{2}  ( x - \log \alpha)^2 + \alpha ( 1 -
  \log \alpha)$.
\end{lemma}

\begin{proof}
  Let $g (x) = e^x - \alpha x - \frac{\alpha}{2} ( x - \log \alpha)^2 - \alpha ( 1 - \log \alpha)$. 
 We have $g' ( x) =e^x - \alpha - \alpha ( x - \log \alpha)$, $g''(x) = e^x - \alpha$
  and $g^{(3)} (x) = e^x$. So, $g ( \log \alpha) = g' ( \log \alpha) =
  g'' ( \log \alpha) = 0$ and $g^{( 3)} > 0$. Therefore, $g ( x) > 0$ for $x >
  \log \alpha$.
\end{proof}

\begin{lemma} \label{m2}
 For $\alpha > 0$ and $x > \left( 1 + 1/\sqrt{2}
  \right) \log \alpha$ we have  $e^x
  > \alpha x$.
\end{lemma}

\begin{proof}
  The inequality is trivial when $x \leq 0$. So, we can assume that $x > 0$.
For $0 < \alpha < 1$, we have $e^x > x > \alpha x$. For $1 < \alpha < e$, by Lemma \ref{m1}, we have $e^x -
    \alpha x > 0$ for $x > \log \alpha$. For $\alpha > e$, the largest root of the polynomial
    $\frac{\alpha}{2}  ( x - \log \alpha)^2 + \alpha ( 1 - \log \alpha)$ is $\log
    \alpha + \sqrt{2 ( \log \alpha - 1)}$ which is smaller than $( 1 +
    1/\sqrt{2} ) \log \alpha$ since $\sqrt{2 ( u - 1)} \leq
    u/\sqrt{2}$ for $u \geq 1$. So, the inequality results from
    Lemma \ref{m1} again. 
\end{proof}

\begin{lemma} \label{m3}
  For $\mu, \psi, x > 0$, if 
$$
x > \frac{1}{2 \sqrt{2}} - \frac{1+\sqrt{2}}{2\sqrt{2}}   \, \frac{\log ( 2
  \psi \mu^2)}{\psi}
$$
then we have  
$$
e^{- \psi x} < \frac{\mu}{\sqrt{x + \frac{1}{2}}} \ .
$$
\end{lemma}

\begin{proof}
We have
  \begin{eqnarray*}
    e^{- \psi \cdot x} < \frac{\mu}{\sqrt{x + \frac{1}{2}}} &
    \Longleftrightarrow & ( x + 1/2 ) \, e^{- 2 \psi
    \cdot x} < \mu^2\\
    & \Longleftrightarrow & ( x + 1/2 ) \, e^{- 2 \psi
    \cdot ( x + 1/2 )} < \mu^2 e^{- \psi}\\
    & \Longleftrightarrow & e^{2 \psi \cdot ( x + 1/2 )} > \frac{x + 1/2}{\mu^2 e^{- \psi}}\\
    & \Longleftrightarrow & e^{2 \psi \cdot ( x + 1/2 )} > \frac{1}{2 \psi \mu^2 e^{- \psi}} \ ,2 \psi \cdot ( x + 1/2 )
  \end{eqnarray*}
  By Lemma \ref{m2}, the last inequality is satisfied as soon as 
$$
2 \psi \cdot
  ( x + 1/2 ) > ( 1 + 1/\sqrt{2} ) \log
  \left( \frac{1}{2 \psi \mu^2 e^{- \psi}}  \right) \ .
$$
Moreover, we have
  \begin{align*}
    &2 \psi \cdot ( x + 1/2 ) > ( 1 + 1/\sqrt{2} ) \log \left ( \frac{1}{2 \psi \mu^2 e^{-
    \psi}}  \right) \Longleftrightarrow  2 \psi \cdot x + \psi > ( 1 + 1/\sqrt{2} ) 
    \log \left( \frac{e^{\psi}}{2 \psi \mu^2} \right)\\
    & \Longleftrightarrow  2 \psi \cdot x + \psi > ( 1 + 1/\sqrt{2} ) \psi - ( 1 + 1/\sqrt{2} ) \log ( 2 \psi \mu^2)\\
    & \Longleftrightarrow  2 \psi \cdot x > \frac{1}{\sqrt{2}} \cdot \psi -
    ( 1 + 1/\sqrt{2} ) \log ( 2 \psi \mu^2)\\
    & \Longleftrightarrow  x > \frac{1}{2 \sqrt{2}} - \frac{ 1 + 1/\sqrt{2} }{2}  \frac{\log ( 2 \psi \mu^2)}{\psi}
  \end{align*}
  
\end{proof}

\begin{theorem}
  \label{dertheo}Let $z \in \NN$. A sufficient condition for having
  $P_{SN} (z) < P ( z)$ is $z \geq z_0$ with $z_0=\lceil z_0^*\rceil$ being the smallest integer greater or equal
  to
$$ 
 z_0^* = \max \left( \frac{2}{\pi \left( 1 -
     \frac{q}{p} \right)^2}\, , \, \frac{1}{2 \sqrt{2}} - \frac{\left( 1 +
     \frac{1}{\sqrt{2}} \right)}{2}  \frac{\log \left( \frac{2 \psi (
     p)}{\pi}  \right)}{\psi ( p)} \right) 
$$
  where $\psi ( p) = \frac{q}{p} - 1 - \log \left(  \frac{q}{p}
  \right) - \log \left( \frac{1}{4 pq} \right) > 0$.
\end{theorem}

\begin{proof}
  First, note that
  \begin{align*}
    \psi ( p) & =  \frac{q}{p} - 1 - \log \left(  \frac{q}{p} \right) - \log
    \left( \frac{1}{4 p^2}  \frac{p}{q} \right)\\
    & =  2 \left[ \frac{1}{2 p} - 1 - \log \left( \frac{1}{2 p} \right)
    \right ]
  \end{align*}
  So, $\psi ( p) > 0$ and $z_0$ is well defined. Let $z > z_0$. By Lemma
  \ref{boundpnsz} and Corollary \ref{inepnsz} it is enough to prove that
  $$
    \frac{1}{1 - \frac{q}{p}}  \frac{1}{\sqrt{2 \pi z}} e^{- z \left(
    \frac{q}{p} - 1 - \log \frac{q}{p} \right)} + \frac{1}{2} e^{- z
    \left( \frac{q}{p} - 1 - \log \left(  \frac{q}{p} \right) \right)}  < S
    \sqrt{\frac{z}{z + \frac{1}{2}}}  \frac{s^z}{\sqrt{\pi z}} 
  $$
  We have $z \geq z_0 \geq \frac{2}{\pi \left( 1 - \frac{q}{p}
  \right)^2}$, thus $\frac{1}{1 - \frac{q}{p}}  \frac{1}{\sqrt{2 \pi z}}
  \leq \frac{1}{2}$. So, the inequality is satisfied as soon as $e^{-
  z \psi ( p)} < \frac{\left( \frac{1}{\sqrt{\pi}} \right)}{\sqrt{z +
  \frac{1}{2}}}$ and the result follows from Lemma \ref{m3}.
\end{proof}

%
%
%

$$
\begin{array}{|c||c|c|c|c|c|c|c|c|c|c|}
 \hline
 z_0 & 2 & 3 & 4 & 5 & 6 & 7 & 8 & 9 & 10 & 11\\ 
 \hline
 q\geq & 0.000 & 0.232 & 0.305 & 0.342 & 0.365 &0.381 &0.393 & 0.401 & 0.409 & 0.415\\ 
 \hline
\end{array}
$$
\smallskip
\centerline{\footnotesize{\textbf{Table 5. Sharp values}}}


\pagebreak

\section{Tables for $P(z,\kappa)$.}

For complete Satoshi Tables see \cite{GPM}.
%

{\small
$$
\begin{array}{|c||c|c|c|c|c|c|c|c|c|c|c|c|c|c|c|}
\hline
\kappa \backslash q & 0.02 & 0.04 & 0.06 & 0.08 & 0.1 & 0.12 & 0.14 & 0.16 & 0.18 & 0.2 & 0.22 & 0.24 & 0.26\\ \hline \hline
0.1 & 0 & 0.01 & 0.03 & 0.09 & 0.18 & 0.33 & 0.55 & 0.88 & 1.34 & 1.96 & 2.78 & 3.87 & 5.27\\ \hline
0.2 & 0 & 0.01 & 0.05 & 0.11 & 0.23 & 0.42 & 0.71 & 1.12 & 1.68 & 2.44 & 3.44 & 4.74 & 6.39\\ \hline
0.3 & 0 & 0.02 & 0.06 & 0.15 & 0.3 & 0.55 & 0.91 & 1.42 & 2.11 & 3.04 & 4.24 & 5.77 & 7.7\\ \hline
0.4 & 0 & 0.02 & 0.08 & 0.19 & 0.39 & 0.69 & 1.14 & 1.77 & 2.62 & 3.74 & 5.17 & 6.98 & 9.22\\ \hline
0.5 & 0 & 0.03 & 0.1 & 0.24 & 0.49 & 0.87 & 1.43 & 2.2 & 3.22 & 4.56 & 6.25 & 8.36 & 10.93\\ \hline
0.6 & 0 & 0.04 & 0.13 & 0.31 & 0.61 & 1.08 & 1.76 & 2.69 & 3.92 & 5.49 & 7.47 & 9.9 & 12.83\\ \hline
0.7 & 0.01 & 0.05 & 0.16 & 0.38 & 0.75 & 1.33 & 2.14 & 3.25 & 4.7 & 6.54 & 8.82 & 11.59 & 14.89\\ \hline
0.8 & 0.01 & 0.06 & 0.19 & 0.46 & 0.92 & 1.61 & 2.58 & 3.88 & 5.57 & 7.7 & 10.3 & 13.42 & 17.11\\ \hline
0.9 & 0.01 & 0.07 & 0.24 & 0.56 & 1.11 & 1.92 & 3.06 & 4.58 & 6.53 & 8.96 & 11.9 & 15.39 & 19.45\\ \hline
1 & 0.01 & 0.08 & 0.28 & 0.67 & 1.32 & 2.27 & 3.6 & 5.36 & 7.58 & 10.32 & 13.61 & 17.47 & 21.9\\ \hline
1.1 & 0.01 & 0.1 & 0.34 & 0.8 & 1.55 & 2.66 & 4.19 & 6.2 & 8.71 & 11.78 & 15.42 & 19.64 & 24.44\\ \hline
1.2 & 0.02 & 0.12 & 0.4 & 0.94 & 1.81 & 3.09 & 4.84 & 7.1 & 9.92 & 13.32 & 17.32 & 21.91 & 27.05\\ \hline
1.3 & 0.02 & 0.14 & 0.47 & 1.09 & 2.1 & 3.55 & 5.53 & 8.07 & 11.2 & 14.95 & 19.3 & 24.24 & 29.72\\ \hline
1.4 & 0.02 & 0.16 & 0.54 & 1.26 & 2.4 & 4.06 & 6.27 & 9.1 & 12.55 & 16.64 & 21.34 & 26.62 & 32.41\\ \hline
1.5 & 0.02 & 0.19 & 0.62 & 1.44 & 2.74 & 4.59 & 7.06 & 10.18 & 13.96 & 18.39 & 23.44 & 29.04 & 35.12\\ \hline
1.6 & 0.03 & 0.22 & 0.71 & 1.64 & 3.1 & 5.17 & 7.9 & 11.32 & 15.43 & 20.2 & 25.58 & 31.49 & 37.83\\ \hline
1.7 & 0.03 & 0.25 & 0.81 & 1.85 & 3.48 & 5.78 & 8.78 & 12.51 & 16.95 & 22.06 & 27.76 & 33.96 & 40.53\\ \hline
1.8 & 0.04 & 0.28 & 0.91 & 2.08 & 3.89 & 6.42 & 9.7 & 13.75 & 18.52 & 23.95 & 29.96 & 36.42 & 43.2\\ \hline
1.9 & 0.04 & 0.32 & 1.03 & 2.33 & 4.32 & 7.1 & 10.67 & 15.03 & 20.13 & 25.88 & 32.18 & 38.88 & 45.84\\ \hline
2 & 0.05 & 0.36 & 1.15 & 2.58 & 4.78 & 7.8 & 11.67 & 16.35 & 21.77 & 27.83 & 34.4 & 41.32 & 48.43\\ \hline
2.1 & 0.05 & 0.4 & 1.28 & 2.86 & 5.26 & 8.54 & 12.71 & 17.7 & 23.44 & 29.8 & 36.62 & 43.74 & 50.96\\ \hline
2.2 & 0.06 & 0.44 & 1.41 & 3.15 & 5.77 & 9.31 & 13.78 & 19.09 & 25.14 & 31.78 & 38.84 & 46.12 & 53.43\\ \hline
2.3 & 0.07 & 0.49 & 1.56 & 3.46 & 6.3 & 10.11 & 14.88 & 20.51 & 26.86 & 33.77 & 41.04 & 48.46 & 55.84\\ \hline
2.4 & 0.07 & 0.54 & 1.71 & 3.78 & 6.85 & 10.94 & 16.01 & 21.95 & 28.59 & 35.75 & 43.21 & 50.76 & 58.17\\ \hline
2.5 & 0.08 & 0.6 & 1.87 & 4.11 & 7.42 & 11.79 & 17.17 & 23.41 & 30.34 & 37.73 & 45.36 & 53 & 60.43\\ \hline
2.6 & 0.09 & 0.65 & 2.04 & 4.46 & 8.01 & 12.67 & 18.35 & 24.89 & 32.09 & 39.7 & 47.48 & 55.19 & 62.6\\ \hline
2.7 & 0.1 & 0.71 & 2.22 & 4.83 & 8.62 & 13.57 & 19.56 & 26.39 & 33.84 & 41.65 & 49.56 & 57.32 & 64.7\\ \hline
2.8 & 0.11 & 0.78 & 2.41 & 5.21 & 9.26 & 14.49 & 20.78 & 27.9 & 35.59 & 43.59 & 51.6 & 59.38 & 66.71\\ \hline
2.9 & 0.12 & 0.85 & 2.6 & 5.6 & 9.91 & 15.44 & 22.02 & 29.42 & 37.34 & 45.5 & 53.6 & 61.39 & 68.64\\ \hline
3 & 0.13 & 0.92 & 2.81 & 6.01 & 10.58 & 16.4 & 23.28 & 30.94 & 39.08 & 47.38 & 55.55 & 63.32 & 70.49\\ \hline
3.1 & 0.14 & 0.99 & 3.02 & 6.44 & 11.27 & 17.38 & 24.55 & 32.47 & 40.81 & 49.24 & 57.45 & 65.19 & 72.25\\ \hline
3.2 & 0.15 & 1.07 & 3.24 & 6.87 & 11.97 & 18.38 & 25.83 & 34 & 42.52 & 51.06 & 59.31 & 67 & 73.93\\ \hline
3.3 & 0.16 & 1.15 & 3.47 & 7.32 & 12.69 & 19.39 & 27.12 & 35.52 & 44.22 & 52.85 & 61.11 & 68.73 & 75.53\\ \hline
3.4 & 0.17 & 1.23 & 3.7 & 7.78 & 13.43 & 20.42 & 28.42 & 37.05 & 45.9 & 54.61 & 62.86 & 70.39 & 77.05\\ \hline
3.5 & 0.19 & 1.32 & 3.95 & 8.26 & 14.18 & 21.46 & 29.73 & 38.56 & 47.56 & 56.32 & 64.55 & 71.99 & 78.5\\
\hline
\end{array}
$$
}
\smallskip
\centerline{\footnotesize{\textbf{Table 6. $P(3, \kappa)$ ($z=3$) for different values of $\kappa$ and $q$ in $\%$.}}}

%

{\small
$$
\begin{array}{|c||c|c|c|c|c|c|c|c|c|c|c|c|c|c|c|}
\hline
\kappa \backslash q & 0.02 & 0.04 & 0.06 & 0.08 & 0.1 & 0.12 & 0.14 & 0.16 & 0.18 & 0.2 & 0.22 & 0.24 & 0.26\\ \hline \hline
0.1 & 0 & 0 & 0 & 0 & 0 & 0 & 0 & 0.01 & 0.02 & 0.04 & 0.08 & 0.15 & 0.28\\ \hline
0.2 & 0 & 0 & 0 & 0 & 0 & 0 & 0.01 & 0.01 & 0.03 & 0.06 & 0.12 & 0.23 & 0.41\\ \hline
0.3 & 0 & 0 & 0 & 0 & 0 & 0 & 0.01 & 0.02 & 0.05 & 0.09 & 0.18 & 0.34 & 0.6\\ \hline
0.4 & 0 & 0 & 0 & 0 & 0 & 0.01 & 0.01 & 0.03 & 0.07 & 0.15 & 0.28 & 0.51 & 0.88\\ \hline
0.5 & 0 & 0 & 0 & 0 & 0 & 0.01 & 0.02 & 0.05 & 0.11 & 0.23 & 0.42 & 0.75 & 1.28\\ \hline
0.6 & 0 & 0 & 0 & 0 & 0 & 0.01 & 0.04 & 0.08 & 0.17 & 0.34 & 0.63 & 1.1 & 1.84\\ \hline
0.7 & 0 & 0 & 0 & 0 & 0.01 & 0.02 & 0.06 & 0.13 & 0.26 & 0.51 & 0.91 & 1.57 & 2.57\\ \hline
0.8 & 0 & 0 & 0 & 0 & 0.01 & 0.03 & 0.08 & 0.19 & 0.39 & 0.73 & 1.3 & 2.19 & 3.53\\ \hline
0.9 & 0 & 0 & 0 & 0 & 0.02 & 0.05 & 0.12 & 0.28 & 0.55 & 1.03 & 1.81 & 2.99 & 4.73\\ \hline
1 & 0 & 0 & 0 & 0.01 & 0.02 & 0.07 & 0.18 & 0.39 & 0.78 & 1.43 & 2.45 & 3.99 & 6.19\\ \hline
1.1 & 0 & 0 & 0 & 0.01 & 0.04 & 0.1 & 0.25 & 0.54 & 1.06 & 1.92 & 3.25 & 5.2 & 7.93\\ \hline
1.2 & 0 & 0 & 0 & 0.01 & 0.05 & 0.14 & 0.35 & 0.74 & 1.42 & 2.53 & 4.21 & 6.63 & 9.94\\ \hline
1.3 & 0 & 0 & 0 & 0.02 & 0.07 & 0.2 & 0.47 & 0.98 & 1.86 & 3.26 & 5.35 & 8.29 & 12.23\\ \hline
1.4 & 0 & 0 & 0 & 0.03 & 0.09 & 0.26 & 0.62 & 1.28 & 2.39 & 4.14 & 6.68 & 10.19 & 14.79\\ \hline
1.5 & 0 & 0 & 0.01 & 0.03 & 0.12 & 0.34 & 0.8 & 1.64 & 3.02 & 5.15 & 8.19 & 12.3 & 17.58\\ \hline
1.6 & 0 & 0 & 0.01 & 0.05 & 0.16 & 0.45 & 1.02 & 2.06 & 3.76 & 6.31 & 9.89 & 14.63 & 20.59\\ \hline
1.7 & 0 & 0 & 0.01 & 0.06 & 0.21 & 0.57 & 1.29 & 2.56 & 4.6 & 7.62 & 11.77 & 17.16 & 23.78\\ \hline
1.8 & 0 & 0 & 0.02 & 0.08 & 0.27 & 0.71 & 1.6 & 3.14 & 5.56 & 9.07 & 13.82 & 19.86 & 27.13\\ \hline
1.9 & 0 & 0 & 0.02 & 0.1 & 0.34 & 0.89 & 1.96 & 3.79 & 6.63 & 10.67 & 16.04 & 22.72 & 30.59\\ \hline
2 & 0 & 0 & 0.03 & 0.12 & 0.42 & 1.09 & 2.37 & 4.53 & 7.82 & 12.42 & 18.4 & 25.71 & 34.14\\ \hline
2.1 & 0 & 0 & 0.03 & 0.15 & 0.51 & 1.32 & 2.83 & 5.35 & 9.12 & 14.29 & 20.9 & 28.81 & 37.73\\ \hline
2.2 & 0 & 0 & 0.04 & 0.19 & 0.62 & 1.58 & 3.36 & 6.26 & 10.54 & 16.29 & 23.51 & 31.98 & 41.34\\ \hline
2.3 & 0 & 0 & 0.05 & 0.23 & 0.75 & 1.88 & 3.95 & 7.26 & 12.06 & 18.41 & 26.23 & 35.21 & 44.94\\ \hline
2.4 & 0 & 0.01 & 0.06 & 0.28 & 0.89 & 2.21 & 4.59 & 8.35 & 13.69 & 20.64 & 29.02 & 38.47 & 48.49\\ \hline
2.5 & 0 & 0.01 & 0.07 & 0.33 & 1.05 & 2.59 & 5.3 & 9.52 & 15.42 & 22.95 & 31.87 & 41.73 & 51.97\\ \hline
2.6 & 0 & 0.01 & 0.09 & 0.4 & 1.24 & 3 & 6.08 & 10.78 & 17.24 & 25.35 & 34.77 & 44.98 & 55.35\\ \hline
2.7 & 0 & 0.01 & 0.1 & 0.47 & 1.44 & 3.45 & 6.92 & 12.12 & 19.15 & 27.81 & 37.69 & 48.19 & 58.63\\ \hline
2.8 & 0 & 0.01 & 0.12 & 0.55 & 1.67 & 3.95 & 7.82 & 13.54 & 21.14 & 30.33 & 40.62 & 51.34 & 61.78\\ \hline
2.9 & 0 & 0.02 & 0.14 & 0.64 & 1.92 & 4.49 & 8.79 & 15.04 & 23.19 & 32.89 & 43.54 & 54.42 & 64.8\\ \hline
3 & 0 & 0.02 & 0.17 & 0.74 & 2.2 & 5.08 & 9.82 & 16.6 & 25.31 & 35.48 & 46.44 & 57.41 & 67.66\\ \hline
3.1 & 0 & 0.02 & 0.19 & 0.85 & 2.5 & 5.71 & 10.91 & 18.24 & 27.47 & 38.08 & 49.29 & 60.3 & 70.38\\ \hline
3.2 & 0 & 0.03 & 0.22 & 0.97 & 2.83 & 6.39 & 12.06 & 19.93 & 29.68 & 40.68 & 52.1 & 63.09 & 72.94\\ \hline
3.3 & 0 & 0.03 & 0.26 & 1.11 & 3.18 & 7.11 & 13.27 & 21.68 & 31.93 & 43.28 & 54.84 & 65.75 & 75.33\\ \hline
3.4 & 0 & 0.03 & 0.3 & 1.25 & 3.57 & 7.88 & 14.54 & 23.48 & 34.2 & 45.86 & 57.52 & 68.3 & 77.57\\ \hline
3.5 & 0 & 0.04 & 0.34 & 1.41 & 3.98 & 8.69 & 15.86 & 25.33 & 36.48 & 48.41 & 60.11 & 70.72 & 79.66\\
\hline
\end{array}
$$
}
\medskip
\centerline{\footnotesize{\textbf{Table 7. $P(6, \kappa)$ ($z=6$) for different values of $\kappa$ and $q$ in $\%$.}}}


\footnotesize{\thanks{Acknowledgements: {We are grateful to N. Emerson for his comments and remarks.}}}

\end{document}